\documentclass[a4paper,UKenglish,cleveref, autoref, thm-restate]{lipics-v2021}

\newcommand{\NP}{\mathrm{NP}}
\newcommand{\coNP}{\mathrm{coNP}}

\newcommand{\Z}{\mathbb Z}
\newcommand{\R}{\mathbb R}
\newcommand{\Zinfty}{\mathbb Z^{\pm \infty}}
\newcommand{\Rinfty}{\mathbb R^{\pm \infty}}
\newcommand{\re}{\rightarrow}

\newcommand{\VMin}{V_{\mathrm{Min}}}
\newcommand{\VMax}{V_{\mathrm{Max}}}

\newcommand{\val}{\mathrm{val}}

\newcommand{\MP}{\mathrm{MP}}
\newcommand{\Enp}{\mathrm{En}^+}
\newcommand{\Enm}{\mathrm{En}^-}

\newcommand{\const}{\models}

\newcommand{\game}{\mathcal G}
\newcommand{\summ}{\mathrm{sum}}

\newcommand{\Min}{\mathrm{Min}}
\newcommand{\Max}{\mathrm{Max}}

\newcommand{\ext}{\mathrm{ext}}
\newcommand{\alt}{\mathrm{alt}}

\title{The GKK Algorithm is the Fastest over Simple Mean-Payoff Games}

\author{Pierre Ohlmann}{IRIF, Université de Paris, France \and \url{https://www.irif.fr/~ohlmann/} }{ohlmann@irif.fr}{https://orcid.org/0000-0002-4685-5253}{}

\authorrunning{P. Ohlmann} 

\Copyright{Pierre Ohlmann} 



\category{} 

\relatedversion{} 

\nolinenumbers 

\EventEditors{John Q. Open and Joan R. Access}
\EventNoEds{2}
\EventLongTitle{42nd Conference on Very Important Topics (CVIT 2016)}
\EventShortTitle{CVIT 2016}
\EventAcronym{CVIT}
\EventYear{2016}
\EventDate{December 24--27, 2016}
\EventLocation{Little Whinging, United Kingdom}
\EventLogo{}
\SeriesVolume{42}
\ArticleNo{23}

\begin{document}

\maketitle

\begin{abstract}
We study the algorithm of Gurvich, Karzanov and Khachyian (GKK algorithm) when it is ran over mean-payoff games with no simple cycle of weight zero.
We propose a new symmetric analysis, lowering the $O(n^2N)$ upper-bound of Pisaruk on the number of iterations down to $N + E^+ + E^- + 1 \leq nN +1$, where $n$ is the number of vertices, $N$ is the largest absolute value of a weight, and $E^+$ and $E^-$ are respectively the largest finite energy and dual-energy values of the game.
Since each iteration is computed in $O(m)$, this improves on the state of the art pseudopolynomial $O(mnN)$ runtime bound of Brim, Chaloupka, Doyen, Gentilini and Raskin, by taking into account the structure of the game graph.
We complement our result by showing that the analysis of Dorfman, Kaplan and Zwick also applies to the GKK algorithm, which is thus also subject to the state of the art combinatorial runtime bound of $O(m2^{n/2})$.
\end{abstract}

\newpage

\section{Introduction}
\label{sec:introduction}

\paragraph*{Mean-payoff and energy games}

In the games under study, two players, Min and Max, take turns in moving a token over a sinkless finite directed graph whose edges are labelled by (potentially negative) integers, interpreted as payoffs from Min to Max.
In a mean-payoff game, the players aim to optimise the average payoff in the long run.
When playing an energy game, Min and Max optimise the profile upper-bound which takes values in $[0,+\infty]$; in a dual-energy game, the profile lower-bound in $[0,-\infty]$ comes under scrutiny.

These three games are determined~\cite{Martin75}: for each initial vertex $v$, there is a value $x$ such that starting from $v$, the minimiser can ensure an outcome $\leq x$ whereas the maximiser can ensure a least $x$.
They are moreover positionaly determined~\cite{EM79, BFLMS08} which means that the players can achieve the optimal value even when restricted to strategies with no memory.
We refer to Figure~\ref{fig:mean_payoff2} for a complete example.

\begin{figure}[h]
\begin{center}
\includegraphics[width=0.7\linewidth]{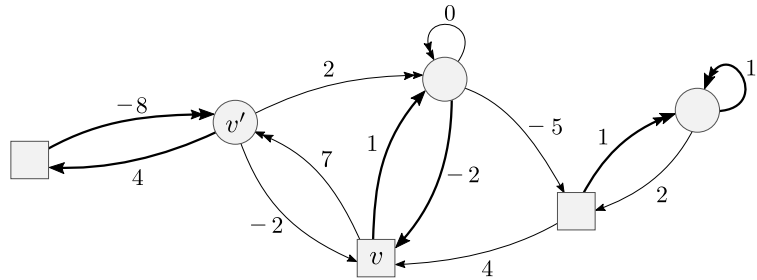}
\end{center}
\caption[caption]{Example of a game; circles and squares represent vertices which respectively belong to Min and Max.
Mean-payoff values from left to right are $-2,-2,-\frac 1 2, - \frac 1 2,1$ and $1$, and mean-payoff-optimal positional strategies for both players are identified in bold.
Energy values are $0,2,9,0,\infty$ and $\infty$, and energy-optimal strategies are given by arrows with double heads.
Dual energy values are $-\infty,-\infty,-\infty,-\infty,0$ and $0$; the bold strategy also gives an optimal strategy in the dual-energy game.
}
\label{fig:mean_payoff2}
\end{figure}

In this paper, we are interested in solving the threshold problem for mean-payoff games: given a game and an initial vertex, decide whether its value is $\leq 0$.
As a consequence of positional determinacy, the mean-payoff value of a vertex is non-positive if and only if the energy value is finite~\cite{BCDGR11}.
In fact, all state of the art algorithms~\cite{BCDGR11,DKZ19,BV05} for the threshold problem -- further discussed below -- actually go through computing the energy values.
The best algorithms for the more general problems of computing the exact values or synthesising optimal strategies in the mean-payoff game also rely on solving many auxiliary energy games~\cite{CR17}.

Positional strategies achieving positive or non-positive values can be checked in polynomial time, and therefore the problem belongs to $\NP \cap \coNP$.
Despite numerous efforts, no polynomial algorithm is known.
Mean-payoff games are known~\cite{Puri95} to be more general than parity games~\cite{EJ91, Mostowski:1991} which enjoy a similar complexity status but were recently shown to be solvable in quasipolynomial time~\cite{CJKLS17}.
It is however unlikely that algorithms for solving parity games in polynomial time generalise to mean-payoff games~\cite{FGO18}.

We use $n$ for the number of vertices, $m \geq n$ for the number of edges, and $N$ for the maximal absolute value of a weight.
We will say that a runtime bound (or an algorithm) is combinatorial if it does not depend on $N$, and that it is pseudopolynomial if it is polynomial in $n$ and $N$.

Although such a terminology was not introduced at that time, the first algorithm for solving energy games is due to Gurvich, Karzanov and Khachyian~\cite{GKK88}.
They used such an algorithm, which we will call the GKK algorithm, as a subroutine in a dichotomy for computing the values in the mean-payoff game.
The GKK algorithm is based on iterating \emph{potential transformations}, each of which require $O(m)$ operations.
From their proof of termination, one can immediately extract an upper bound of $O(n2^n)$ on the number of iterations, which is easily improved to $O(2^n)$ with a slightly refined analysis.
The results of Pisaruk~\cite{Pisaruk99} in a more general setting imply a pseudopolynomial bound of $O(n^2 N)$ on the number of iterations of the GKK algorithm, aligning its worst case runtime bound with that of Zwick and Paterson~\cite{ZP96}.

The current state-of-the-art combinatorial algorithm is the  strategy improvement algorithm of Bjorklund and Vorobiov~\cite{BV05} with randomised runtime $\min(O(mn^2N),2^{O(\sqrt{n \log n})})$.
The pseudopolynomial bound was later improved by Brim, Chaloupka, Doyen, Gentilini and Raskin~\cite{BCDGR11} by reduction to energy games to $O(mnN)$ with a deterministic value iteration algorithm.
This technique was recently refined by Dorfman, Kaplan and Zwick~\cite{DKZ19} who proposed an acceleration of the algorithm which runs in time $O(\min(mnN,m2^{n/2}))$.
Currently, this is the best known deterministic algorithm for the threshold problem, both in terms of combinatorial and pseudopolynomial bounds; in particular, no deterministic subexponential algorithm is known to this day.

\paragraph*{Our contribution}

We propose to analyse the GKK algorithm when it is ran over a \emph{simple} mean-payoff game, meaning, one which has no simple cycle of weight zero.
Simple mean-payoff games arise directly when translating from parity games; moreover one can reduce in general to a simple game with a multiplicative blow-up of $n$ on the largest weight $N$.
We give a completely symmetric presentation of the GKK algorithm in this case, and a novel symmetric analysis based on energy and dual-energy values.

Our main result is a novel bound of $N + E^+ + E^- +1$ on the number of iterations of the GKK algorithm over simple games, where $E^+$ and $E^-$ are respectively the maximal finite energy and dual-energy values of a vertex.
This quantity is always smaller than $nN+1$, and therefore the GKK algorithm is at least as efficient in this case as the state of the art value iteration algorithms~\cite{BCDGR11,DKZ19}.

In practice however, $N + E^+ + E^-$ may be much smaller than $nN$; for instance in the game of Figure~\ref{fig:mean_payoff2}, we have $N+E^++E^-=8+9+0=17$ whereas $nN=48$. It is very easy to forge examples where the difference is much higher; we believe that for many natural classes of games it holds that $N+E^+ + E^- = o(nN)$.
Moreover, the value iteration algorithms rely on using $nN$ as a threshold beyond which energy values are considered to be infinite, and therefore they often display runtime $\Omega(nN)$ when there are vertices with positive mean-payoff value.
Our result indicates that the GKK algorithm avoids this drawback, all the while retaining (and often improving, as explained above) the state of the art pseudopolynomial runtime bound, at least for simple games.

We complement our main bound by showing that the analysis of~\cite{DKZ19} can also be applied to the GKK algorithm, establishing a combinatorial $O(2^{n/2})$ bound on the number of iterations.
Hence the GKK algorithm also matches the combinatorial state of the art for deterministic algorithms (here, the fact that simple arenas are used is not a restriction, since the reduction only blows up the size of the weights).
We also believe that the analysis of Dorfman, Kaplan and Zwick is conceptually simpler (and completely symmetric) when instantiated to the GKK algorithm.

In Section~\ref{sec:preliminaries}, we formally introduce the necessary definitions and concepts.
Section~\ref{sec:gkk_algorithm} presents the GKK algorithm over simple games, and Section~\ref{sec:pseudopolynomial} provides the novel pseudopolynomial bound.
In Section~\ref{sec:combinatorial_bound}, the combinatorial bound is derived.

\section{Preliminaries}\label{sec:preliminaries}

In this preliminary section, we introduce mean-payoff and energy games, potential reductions, and discuss simple games.

\paragraph*{Mean-payoff and energy games}

In this paper, a game is a tuple $\game=(G,w,\VMin, \VMax)$, where $G=(V,E)$ is a finite directed graph with no sink, $w:V \to \Z$ is a labelling of its edges by integer weights, and $\VMin,\VMax$ is a partition of $V$.
As in the introduction, we use $n,m$ and $N$ respectively for $|V|,|E|$ and $\max_{e} |w(e)|$; we say that vertices in $\VMin$ belong to Min while those in $\VMax$ belong to Max.
We now fix a game $\game = (G,w,\VMin,\VMax)$.

A path is a (possibly empty, possibly infinite) sequence of edges $\pi = e_0e_1 \dots$ with matching endpoints: if $e_{i+1}=v_{i+1}v_{i+2}$ is defined then its first component $v_{i+1}$ matches the second component of $e_i$.
For convenience, we often write $v_0 \re v_1 \re v_2 \re \dots$ for the path $e_0e_1 \dots = (v_0v_1)(v_1v_2) \dots$.
Given a finite or infinite path $\pi=e_0 e_1 \dots$ we let $w(\pi)=w(e_0)w(e_1) \dots$ denote the sequence of weights appearing on $\pi$.
The sum of a finite path $\pi$ is the sum of the weights appearing on it, we denote it by $\summ(\pi)$.

Given a finite or infinite path $\pi=e_0e_1 \dots = v_0 \re v_1 \re \dots$ and an integer $k \geq 0$, we let $\pi_{< k} = e_0 e_1 \dots e_{k-1} = v_0 \re v_1 \re \dots \re v_k$, and we let $\pi_{\leq k} = \pi_{< k+1}$.
Note that $\pi_{<0}$ is the empty path, and that $\pi_{<k}$ has length $k$ in general: it belongs to $E^k$.
We say that $\pi$ starts in $v_0$, and when it is finite and of length $k$ that it ends in $v_{k}$.
By convention, the empty path starts and ends in all vertices.
A cycle is a finite path which starts and ends in the same vertex.
A finite path $v_0 \re v_1 \re \dots \re v_k$ is simple if there is no repetition in $v_0,v_1,\dots,v_{k-1}$; note that a cycle may be simple.
We let $\Pi_v^\omega$ denote the set of infinite paths starting in $v$.

We use $\Rinfty$ and $\Zinfty$ to denote respectively $\R \cup \{-\infty, +\infty\}$ and $\Z \cup \{-\infty, \infty\}$.
A valuation is a map $\val: \Z^\omega \to \Rinfty$ which assigns a potentially infinite real number to infinite sequences of weights.
The three valuations which are studied in this paper are the mean-payoff, energy, and dual-energy valuations, respectively given by
\[
\begin{array}{lclcl}
\MP(w_0w_1 \dots) &=& \limsup_k \frac 1 k \sum_{i=0}^{k-1} w_i & \in & \R \\
\Enp(w_0w_1 \dots) &=& \sup_k \sum_{i=0}^{k-1} w_i & \in & [0,\infty] \\
\Enm(w_0w_1 \dots) &=& \inf_k \sum_{i=0}^{k-1} w_i & \in & [-\infty,0].
\end{array}
\]
A strategy for Min is a map $\sigma: \VMin \to E$ such that for all $v \in \VMin$, it holds that $\sigma(v)$ is an edge outgoing from $v$.
We say that a (finite or infinite) path $\pi = e_0e_1 \dots = v_0 \re v_1 \re \dots$ is consistent with $\sigma$ if whenever $e_i=v_iv_{i+1}$ is defined and such that $v_i \in \VMin$, it holds that $e_i = \sigma(v_i)$.
We write in this case $\pi \const \sigma$.
Strategies for Max are defined similarly and written $\tau : \VMax \to E$.
Paths consistent with Max strategies are defined analogously and also denoted by $\pi \const \tau$.

\begin{theorem}[\cite{EM79,BFLMS08}]\label{thm:positionality}
For each $\val \in \{\MP,\Enp,\Enm\}$, there exist strategies $\sigma_0$ for Min and $\tau_0$ for Max such that for all $v \in V$ we have
\[
\sup_{\pi \const \sigma_0} \val(w(\pi)) = \inf_{\sigma} \sup_{\pi \const \sigma} \val(w(\pi)) = \sup_{\tau} \inf_{\pi \const \tau} \val(w(\pi)) = \inf_{\pi \const \tau_0} \val(w(\pi)),
\]
where $\sigma, \tau$ and $\pi$ respectively range over strategies for Min, strategies for Max, and infinite paths from $v$.
\end{theorem}

The quantity defined by the equilibrium above is called the value of $v$ in the $\val$ game, and we denote it by $\val_{\game}(v) \in \Rinfty$; the strategies $\sigma_0$ and $\tau_0$ are called $\val$-optimal, note that they do not depend on $v$.
The following result relates the values in the mean-payoff and energy games; this direct consequence of Theorem~\ref{thm:positionality} was first stated in~\cite{BCDGR11}.

\begin{corollary}[\cite{BCDGR11}]\label{cor:mean_payoffs_and_energies}
For all $v \in V$ it holds that
\[
\MP_\game(v) \leq 0 \iff \Enp_\game(v) < \infty \iff \Enp_\game(v) \leq (n-1)N,
\]
and likewise,
\[
\MP_\game(v) \geq 0 \iff \Enm_\game(v) > -\infty \iff \Enm_\game(v) \geq -(n-1)N.
\]
\end{corollary}

Therefore computing $\Enp$ values of the games is harder than the threshold problem.
As explained in the introduction, all state-of-the-art algorithms for the threshold problem actually compute $\Enp$ values, and so does the GKK algorithm (in fact, it even computes $\Enm$ values, while algorithms of~\cite{BV05,BCDGR11,DKZ19} do not).
This shifts our focus from mean-payoff to energy games.

\paragraph*{Potential reductions}

Fix a game $\game=(G=(V,E),w,\VMin,\VMax)$.
A potential is a map $\phi:V \to \Z$.
Potentials are partially ordered coordinatewise.
Given an edge $e=vv' \in E$, we define its $\phi$-modified weight to be
\[
w_\phi(e)=w(e) + \phi(v') - \phi(v).
\]
The $\phi$-modified game $\game_\phi$ is simply the game $(G,w_\phi,\VMin,\VMax)$; informally, all weights are replaced by the modified weights.
Note that the underlying graph does not change, in particular paths in $\game$ and $\game_\phi$ are the same.
Observe that for a finite path $\pi= v_0 \re v_1 \re \dots \re v_k$, its sum in $\game_\phi$ is given by
\[
\summ_\phi(v) = \summ(v) - \phi(v_0) + \phi(v_k).
\]
We let $0$ denote the constant zero potential; note that $\game^0=\game$.
For convenience, we use $\val_{\phi}$ to denote $\val_{\game_{\phi}}$ for $\val \in \{\MP,\Enp,\Enm\}$.
Since $\game$ is always fixed, we thus write $\val_0$ for $\val_\game$.

Moving from $\game$ to $\game_\phi$ for a given potential $\phi$ is called a potential reduction; these were introduced by Gallai~\cite{Gallai58} for studying network related problems such as shortest-paths problems.
In the context of mean-payoff or energy games, they were introduced in~\cite{GKK88} and later rediscovered numerous times.
The result below describes the effect of potential reductions over mean-payoff and energy values.

\begin{theorem}\label{thm:potential_reduction_theorem}
\begin{itemize}
\item For any potential $\phi$ we have $\MP_0 = \MP_\phi$ over $V$.
\item If $\phi$ satisfies $0 \leq \phi \leq \Enp_0$, then it holds that $\Enp_0 = \phi + \Enp_\phi$ over $V$.
\item If $\phi$ satisfies $\Enm_0 \leq \phi \leq 0 $, then it holds that $\Enm_0 = \phi + \Enm_\phi$ over $V$.
\end{itemize}
\end{theorem}

A full proof is given in Appendix~\ref{app:potential_reduction_theorem} for completeness.
The second item is illustrated in Figure~\ref{fig:potential_reduction_theorem}.
We say that a potential $\phi$ is positively safe if it satisfies the hypothesis of the second item, $0 \leq \phi \leq \Enp_0.$

\begin{figure}[h]
\begin{center}
\includegraphics[width=0.6\linewidth]{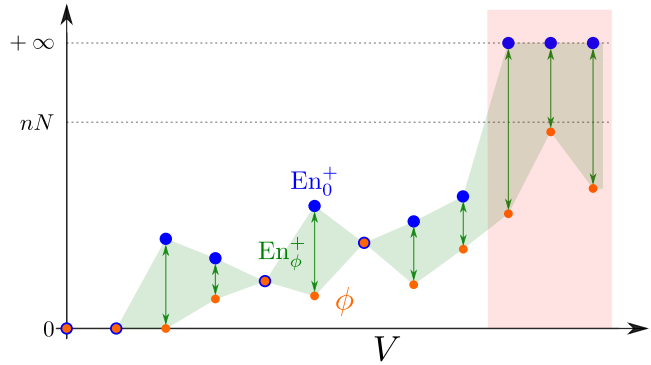}
\end{center}
\caption{An illustration of the second item in Theorem~\ref{thm:potential_reduction_theorem}. For vertices on the right, energy values in both games are $\infty$.}
\label{fig:potential_reduction_theorem}
\end{figure}

Note that potential reductions are invariant under shifts: we have $\game_\phi = \game_{\phi+c}$ if $c$ is a constant potential.
For convenience, we prefer to work with non-negative potentials, even though our approach will be completely symmetric; one could also work with shift-invariant equivalent classes.

To apply the third item in Theorem~\ref{thm:potential_reduction_theorem}, given a potential $\phi$ we define
\[
\phi^- = \phi - \max \phi \leq 0,
\]
and we say that $\phi$ is negatively safe if $\phi^-$ satisfies the hypothesis of the third item, $\Enm_0 \leq \phi^- \leq 0$.
We say that $\phi$ is bi-safe if it is both positively and negatively safe.

Observe that $(\game_\phi)_{\phi'} = \game_{\phi+\phi'}$: sequential applications of potential reductions correspond to reducing with respect to the sum of the potentials.
The following is easily derived as a consequence of Theorem~\ref{thm:potential_reduction_theorem}.

\begin{lemma}\label{lem:safety_composition}
If $\phi$ is positively (or negatively, or bi-) safe for $\game$, and $\phi'$ is positively (or negatively, or bi-) safe for $\game'$, then $\phi + \phi'$ is positively (or negatively, or bi-) safe for $\game$.
\end{lemma}

\paragraph*{Simple and reduced games}

The lemma above justifies the following approach for computing $\Enp_0$: apply successive positively safe potential reductions $\phi_0,\phi_1,\dots$ until reaching a game whose energy values are only $0$ and $\infty$; then by Theorem~\ref{thm:potential_reduction_theorem} it holds that $\Enp_0 = \phi_0 + \phi_1 + \dots$.
We will present the GKK algorithm as one iterating potential reductions that are actually bi-safe.
For this to hold however, we need to restrict to simple games.

A game is simple if all simple cycles have nonzero sum.
The following result is folklore and states that one may reduce to a simple game at the cost of a linear blow up on $N$.
It holds thanks to the fact that positive mean-payoff values are $\geq 1/n$, which is a well-known consequence of Theorem~\ref{thm:positionality}.

\begin{lemma}\label{lem:lifting_simplicity}
Let $\game=(G,V,\VMin,\VMax)$ be an arbitrary game.
The game $\game'=(G,V,(n+1)w-1,\VMin,\VMax)$ is simple and has the same vertices of positive mean-payoff values as $\game$.
\end{lemma}

As another direct consequence of Theorem~\ref{thm:positionality}, it holds that in a simple game, mean-payoff values of the vertices are $\neq 0$.
Energy and dual energy values in such a game are depicted in Figure~\ref{fig:simple_energies}.
Moreover, sums of cycles are preserved by potential reductions, and therefore if $\game$ is simple then so is $\game_\phi$, whatever the potential $\phi$.

\begin{figure}[h]
\begin{center}
\includegraphics[width=0.55\linewidth]{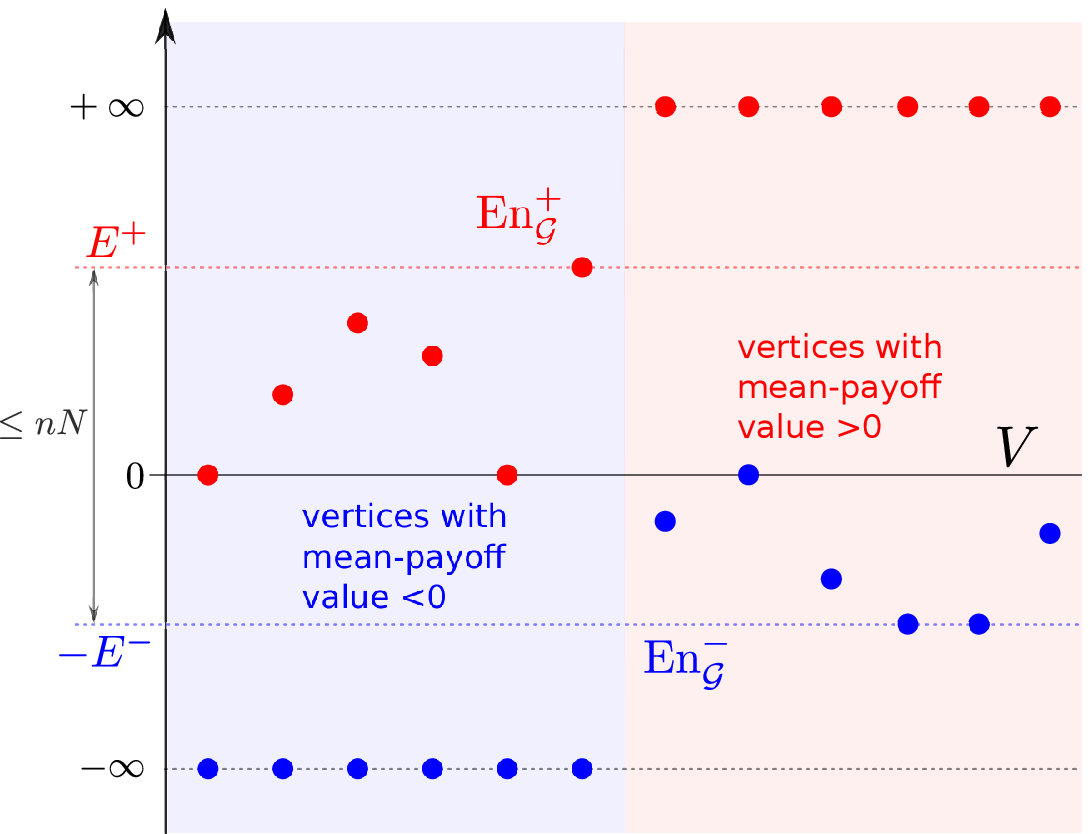}
\end{center}
\caption{Representation of energy and dual energy values when no vertex has mean-payoff value zero; this is always the case for simple arenas.}
\label{fig:simple_energies}
\end{figure}

We say that a simple game is reduced if the vertices are partitioned between $P^*$ and $N^*$ such that
\begin{itemize}
\item vertices in $\VMin \cap N^*$ have a non-positive edge towards $N^*$;
\item all edges outgoing from vertices in $\VMax \cap N^*$ are non-positive and towards $N^*$;
\item vertices in $\VMax \cap P^*$ have a non-negative edge towards $P^*$; and
\item all edges outgoing from vertices in $\VMin \cap P^*$ are non-negative and towards $P^*$.
\end{itemize}
These requirements are illustrated in Figure~\ref{fig:reduced_arena}.

\begin{figure}[h]
\begin{center}
\includegraphics[width=0.55\linewidth]{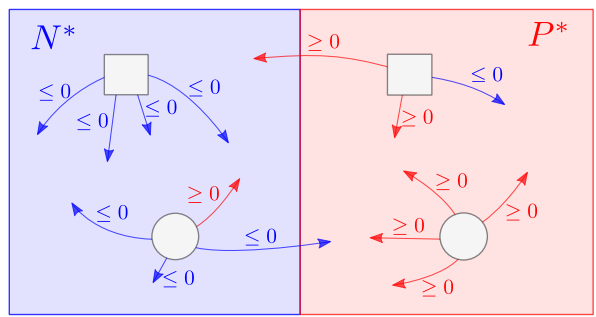}
\end{center}
\caption{A reduced arena. Non-positive edges are represented in blue and non-negative ones in red.}
\label{fig:reduced_arena}
\end{figure}

Intuitively, a reduced game is a simple one in which Min can ensure that no positive edge is ever seen from any vertex of mean-payoff value $<0$, and vice-versa.
We have the following easy result.

\begin{lemma}\label{lem:values_in_reduced_game}
In a reduced game, vertices in $N^*$ have mean-payoff value $<0$ and those in $P^*$ have mean-payoff value $>0$.
Moreover, a simple game is reduced if and only if energy values belong to $\{0,\infty\}$ and dual energy values belong to $\{-\infty,0\}$.
\end{lemma}

\section{The GKK algorithm}\label{sec:gkk_algorithm}

Fix a simple game $\game=(G=(V,E),w,\VMin,\VMax)$.
The GKK algorithm iterates bi-safe potential reductions until a reduced arena is obtained.
The runtime for computing each reduction is $O(m)$, therefore the overall runtime is $O(m\ell)$, where $\ell$ is the number of iterations.
In this section we present how the reduction is performed, and prove that it is bi-safe.
Upper bounds on $\ell$ are the focus of Sections~\ref{sec:pseudopolynomial} and~\ref{sec:combinatorial_bound}.

Each iteration relies on a bipartition of the set of vertices, which is completely symmetric thanks to our simplicity assumption.
Observe that since there are no simple cycles of sum zero in $\game$, any infinite path visits a non-zero weight.
The arena is therefore partitioned into the set of vertices $N^*$ from which Eve can ensure that the first visited non-zero weight is negative, and the set of vertices $P^*$ from which Adam can ensure that the first visited non-zero weight is positive.

Note that the partition $N^*,P^*$ depends only on the signs (and zeroness) of the weights, and not on their precise values.
It is computable in linear time; in a standard terminology which is not formally introduced here, $N^*$ is the Min-attractor to negative edges over non-positive edges.
The GKK algorithm is in fact akin to Zielonka's algorithm for parity games~\cite{Zielonka98}: both are based on computing relevant attractors.

\begin{figure}[ht]
\begin{center}
\includegraphics[width=0.9\linewidth]{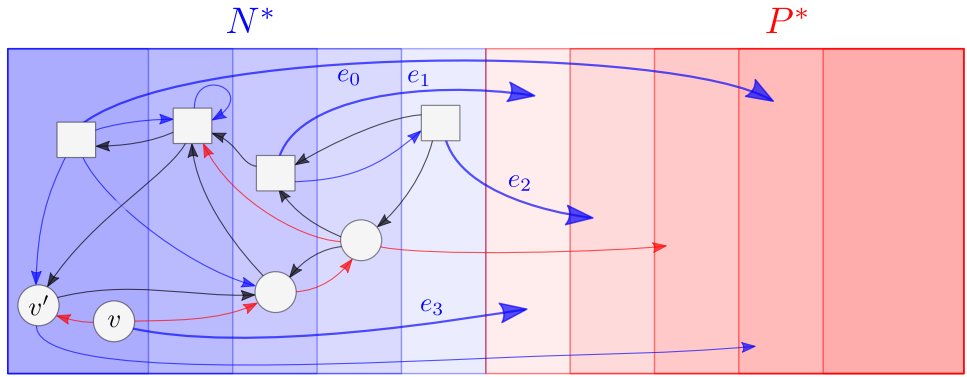}
\end{center}
\caption[caption]{An example of the partition of the vertices into $N^*$ and $P^*$; for clarity, no details are given with respect to $P^*$ where the situation is symmetric.
Blue, black and red arrows respectively represent negative, zero, and positive edges.
The layers depicted in $N^*$ correspond to the Eve-attractor over zero edges to negative ones.\\
With regards to the explanation below: here three edges participate to the maximum defining $\delta_\Max^{-}$ namely $e_0,e_1$ and $e_2$.
Only $e_3$ participates to the maximum defining $\delta_\Min^{-}$; $v'$ has a non-positive edge towards $N^*$ and thus does not belong to $SN$.}
\label{fig:gkk_partition}
\end{figure}

We focus on the point of view of Min, and thus on $N^*$.
By definition, from $N^*$ Min is able to force that a negative edge is seen.
The algorithm computes the worst possible (maximal) negative value that Min can ensure from $N^*$, which we now describe.

Consider a Max vertex $v$ in $N^*$: any edge towards $P^*$ is necessarily negative otherwise $v$ would belong to $P^*$.
Therefore Max may choose to switch to $P^*$, but at the cost of seeing a negative weight.
We let
\[
\delta_{\Max}^- = \max\{w(e) \mid e \in (N^* \cap \VMax) \times P^*\} <0
\]
denote the largest such weight that Max can achieve.
It may be that there is no such edge, in which case we have $\delta_\Max^-=\max \varnothing = -\infty$.

From a Min vertex $v$ in $N^*$ if Min has a non-positive edge towards $N^*$ she can follow this path and avoid to switch to $P^*$.
Otherwise all edges outgoing from $v$ towards $N^*$ are positive, and we let
\[
SN = \{v \in \VMin \cap N^* \mid \forall v' \in N^*, vv' \in E \implies w(vv')>0\}
\]
be the set of Eve vertices in $N^*$ from which she is forced to switch to $P^*$ or see a positive edge.
Note that a vertex $v \in SN$ necessarily has negative outgoing edges, which must therefore point towards $P^*$, otherwise $v$ would not belong to $N^*$.
Therefore we let
\[
\delta_{\Min}^- = \max_{v \in SN} \min\{w(vv') \mid v' \in V\} <0,
\]
and we now put
\[
\delta^- = \max(\delta_\Min^-, \delta_\Max^-) \in [-\infty,0). 
\]
The following result (and the dual one) is crucial for our pseudopolynomial bound.
We prove it now since it refers to the definitions just above.

\begin{lemma}
\label{lem:relevance_deltam}
It holds that $\Enm_0$ takes values $\leq \delta^-$ over $N^*$.
\end{lemma}

\begin{proof}
Consider a positional strategy $\sigma$ for Min which assigns to $v \in (\VMin \cap N^*) \setminus SN$ a non-positive edge towards $N^*$, and to $v \in SN$ an edge of weight $\leq \delta_\Min^-$ (which therefore necessarily leads to $P^*$).
Consider an infinite path $\pi:v_0 \re v_1 \re \dots$ from $v_0 \in N^*$ which is consistent with $\sigma$.

If $\pi$ remains in $N^*$ then all weights are non-positive, and since moreover $\game$ is simple it must be that $\Enm(\pi)=-\infty$.
Otherwise, let $i_0 \geq 0$ be the first index such that $v_{i_0+1} \in P^*$.
If $v_{i_0} \in \VMin$ then necessarily $v_{i_0} \in SN$ and thus $w(v_{i_0}v_{i_0+1}) \leq \delta_\Min^- \leq \delta^-$.
If $v_{i_0} \in \VMax$ then likewise $w(v_{i_0} v_{i_0+1}) \leq \delta_\Max^- \leq \delta^-$.
Since moreover $\pi_{< i_0}$ remains in $N^*$ and is consistent with $\sigma$, it only sees non-positive weights, and therefore $\Enm(\pi) \leq w(v_0v_1) + w(v_1v_2) + \dots + w(v_{i_0}v_{i_0+1}) \leq w(v_{i_0} v_{i_0+1}) \leq \delta$.
\end{proof}

Symmetrically one may define a relevant minimal positive weight for Max from $P^*$ by setting
\[
\delta_{\Min}^+ = \min\{w(e) \mid e \in (P^* \cap \VMin) \times N^*\} \qquad \text{ and } \qquad  
\delta_\Max^+ = \min_{v \in SP} \max\{t \mid v \re t v'\}
\]
where
$
SP = \{v \in \VMax \cap P^* \mid \forall v' \in P^*, vv' \in E \implies w(vv') <0\}$, and then
\[
\delta^+ = \min(\delta_\Min^+,\delta_\Max^+) \in (0,\infty].
\]
The symmetric version of Lemma~\ref{lem:relevance_deltam} states that $\Enp_0$ takes values $\geq \delta^+$ over $P^*$.

We now finally let $\delta = \min(-\delta^-, \delta^+) \in (0,\infty]$.
If $\delta=+\infty$ then $\delta^-=-\infty$ and $\delta^+=+\infty$ which implies that $\game$ is reduced and the iteration stops.
Otherwise we have $\delta>0$ and we consider the non-negative potential given by
\[
\phi(v) = \begin{cases}
\delta &\text{ if } v \in P^* \\
0 & \text{ if } v \in N^*.
\end{cases}
\]
We call it the GKK potential associated to $\game$.
Note that it is symmetric up to shifting by $-\delta/2$, and therefore so is the corresponding potential reduction; it adds $\delta$ to the weight of edges from $N^*$ to $P^*$, removes $\delta$ to the weight of edges from $P^*$ to $N^*$, and leaves other edges unchanged.
Lemma~\ref{lem:relevance_deltam} and the symmetric variant together yield the following result.

\begin{corollary}\label{cor:bi_safe}
The potential $\phi$ is bi-safe.
\end{corollary}

Without the simplicity assumption over $\game$, one has to deal with vertices from which neither player can attract to a weight of corresponding sign.
In~\cite{GKK88}, such vertices are put in $N^*$, and therefore the obtained potential $\phi$ remains positively-safe, but it is no longer negatively safe.
It is thus unclear how to generalise our approach to non-simple games: as it will appear in the next section, bi-safety is crucial to derive our novel upper bound.

\section{Improved pseudopolynomial bound}\label{sec:pseudopolynomial}

Following~\cite{GKK88}, we say that extremal edges of a vertex $v$ are those with minimal weight if $v \in \VMin$ and of maximal weight if $v \in \VMax$.
The extremal weight of $v$ is the weight of its extremal edges. 
We say that a vertex is negative, zero, or positive according to the sign of its extremal weight, and let\footnote{We apologise for the clash in notations with our notation $N$ for the maximal absolute value of a weight; it is easily resolved thanks to context.} $N,Z$ and $P$ denote the corresponding subsets of vertices.
Note that $N \subseteq N^*$ and $P \subseteq P^*$, while $Z$ is split between both.
The following was already observed in~\cite{GKK88}, a proof is given in Appendix~\ref{app:evolution_n_p} for completeness.

\begin{lemma}[\cite{GKK88}]\label{lem:evolution_n_p}
Let $\game' = \game_\phi$ where $\phi$ is the GKK potential associated to $\game$, and let $N'$ and $P'$ respectively denote the sets of negative and positive vertices in $\game'$.
We have $N' \subseteq N$ and $P' \subseteq P$.
\end{lemma}

We now let $\game=\game^0,\game^1,\game^2, \dots$ denote the sequence of games encountered throughout the iteration, inductively defined by $\game^{j+1}=\game^j_{\phi^j}$, where $\phi^j$ is the GKK potential associated to $\game^j$ (if it is defined).
We use obvious notations such as $N^j,P^{*,j}$ or $\delta^j$; in particular, $\game^{j+1}$ is defined if and only if $\delta^j<\infty$.
Given $j$ such that $\game^{j}$ is defined we moreover let $\Delta^{j} = \sum_{j'=0}^{j} \delta^{j'}$ and $\Phi^{j} = \sum_{j'=0}^{j} \phi^{j'}$.
Note that we have $\game^{j+1}=\game^0_{\Phi^{j}}$ for all $j \geq 0$.
The following is a direct consequence of Lemma~\ref{lem:evolution_n_p}.

\begin{corollary}\label{cor:bounds_on_potentials}
For all $j \geq 0$, it holds that $\Phi^j$ takes value $0$ over $N^j$ and $\Delta^j$ over $P^j$.
\end{corollary}

\begin{proof}
Thanks to Lemma~\ref{lem:evolution_n_p} we have $N^0 \supseteq N^1 \supseteq \dots \supseteq N^j$, therefore if $v \in N^j$ then for all $j' \leq j$, $v$ belongs to $N^{j'} \subseteq N^{*,j'}$ and thus $\phi^{j'}(v) =0$; the first result follows.
Likewise, if $v \in P^j$ then for all $j' \leq j$ we have $\phi^{j'}(v) = \delta^{j'}$ therefore $\Phi^j(v)=\Delta^j$.
\end{proof}

With this is hands we are ready to prove the announced result.

\begin{theorem}
\label{thm:pseudopoly_bound}
The iteration terminates in at most $N + E^+ + E^- +1$ steps, where $E^+$ is the maximal finite energy value in $\game$, and $E^-$ is minus the minimal finite dual energy value.
\end{theorem}

The proof is illustrated in Figure~\ref{fig:pseudopoly_upper_bound}.

\begin{figure}[h]
\begin{center}
\includegraphics[width=0.7\linewidth]{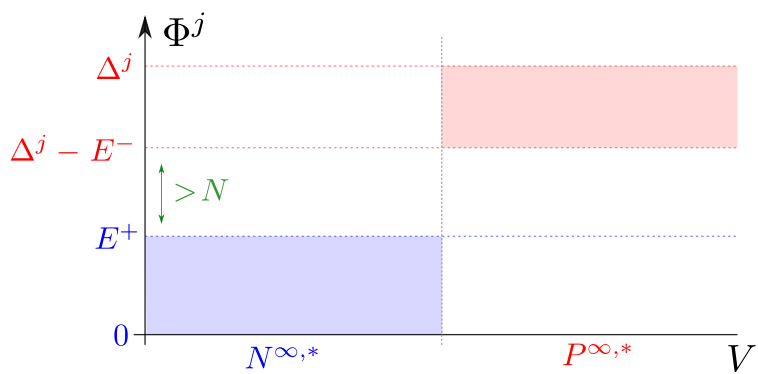}
\end{center}
\caption{An illustration for the proof of Theorem~\ref{thm:pseudopoly_bound}, where $j=N+ E^+ + E^-$.
Since $\Psi^j$ it is positively safe, vertices with finite $\Enp$ value (denoted $N^{\infty,*}$) must be mapped to the blue region, and symmetrically; by our choice of $j$, this implies that edges from $N^{\infty,*}$ to $P^{\infty,*}$ are positive, and those from $P^{\infty,*}$ to $N^{\infty,*}$ are negative, which is key to the proof.}
\label{fig:pseudopoly_upper_bound}
\end{figure}

\begin{proof}
We let $N^{\infty, *}$ and $P^{\infty,*}$ respectively denote the sets of vertices with negative and positive mean-payoff values, which partition $V$.
Since $\Phi^j$ is positively safe by Corollary~\ref{cor:bi_safe} and Lemma~\ref{lem:safety_composition} (and the quantities below are finite), we have thanks to Theorem~\ref{thm:potential_reduction_theorem} for all $j$ that over $v \in N^{\infty,*}$,
\[
\Phi^{j}(v) = \Enp_\game(v) - \Enp_{\game^j}(v) \leq E^+.
\]
Likewise, over $v \in P^{\infty,*}$ we obtain $
\Phi^{j,-}(v) = \Enm_\game(v) - \Enp_{\game^j}(v) \geq -E^-,
$
which rewrites as
\[
\Phi^{j}(v) \geq \Delta^j - E^-.
\]
We now assume that the $j=N + E^+ + E^-$-th iteration is defined, and for contradiction that $\delta^{j} < \infty$.
Note that $\Delta^j \geq j+1$ as a sum of $j+1$ positive integers.
Note that $N^j$ (and symmetrically, $P^j$) is non-empty: if $N^j=\varnothing$ then $P^j=V$ therefore $\delta^j=\infty$. (Intuitively, Max could then ensure that no negative weight is ever seen.)

By Corollary~\ref{cor:bounds_on_potentials}, $\Phi^j$ takes value $0$ over $N^j$ therefore $N^j \subseteq N^{\infty,*}$ thanks to the above since $0< \Delta^j - E^-$ (see Figure~\ref{fig:pseudopoly_upper_bound}; vertices of value zero cannot belong to the red zone). 
Likewise, we have $P^j \subseteq P^{\infty,*}$ since $ \Delta^j>E^+$.

Note that any edge $vv'$ from $N^{\infty,*}$ to $P^{\infty,*}$ has weight
\[
w_{\Phi^j}(vv')=w(vv') + \Phi^j(v') - \Phi^j(v) \geq w(vv') + \Delta^j - E^- - E^+ \geq -N + \Delta^j - E^- - E^+ \geq 1
\]
in $\game^j$.
Likewise, any edge from $P^{\infty,*}$ to $N^{\infty,*}$ has weight $<0$ in $\game^j$, therefore zero edges cannot lead from $N^{\infty,*}$ to $P^{\infty,*}$ or vice-versa.

Now observe that by definition vertices in $N^{j,*}$ have a path to $N^j \subseteq N^{\infty,*}$ comprised only of zero weights in $\game^j$, and therefore it must be that $N^{j,*} \subseteq N^{\infty,*}$.
Similarly, we have $P^{j,*} \subseteq P^{\infty,*}$ and thus the two partitions are equal:
\[
N^{j,*} = N^{\infty,*} \qquad \text{ and } \qquad P^{j,*} = P^{\infty,*}.
\]
Since all edges from $N^{j,*}$ to $P^{j,*}$ are positive, we have $\delta^-=-\infty$.
Likewise $\delta^+=\infty$ and therefore $\delta=\infty$, a contradiction.
\end{proof}

\section{Combinatorial bound}\label{sec:combinatorial_bound}

We now concentrate on establishing the following result.

\begin{theorem}\label{thm:combinatorial_bound}
The number of iterations of the GKK algorithm is $O(2^{n/2})$.
\end{theorem}

It implies that the GKK algorithm also matches the state of the art combinatorial bound of~\cite{DKZ19}; we actually believe that the two algorithms are very similar in essence.
Note that the simplicity assumption can be lifted without loss of generality here: there is no combinatorial blow up in the reduction stated in Lemma~\ref{lem:lifting_simplicity}.
The algorithm of~\cite{DKZ19} has the advantage of benefiting in general from the $O(nmN)$ upper bound inherited from that of~\cite{BCDGR11}, regardless of simplicity.
Inversely, it is not clear whether our improved pseudopolynomial bound holds for the algorithm of~\cite{DKZ19}, even when it is ran over simple arenas.

Our proof of Theorem~\ref{thm:combinatorial_bound} is directly based on that of~\cite{DKZ19}, which we break into two steps.
First, we partition $N^*$ into non-empty layers $A_1,A_2,\dots$ and prove that the sequence $-|A_1|,|A_2|,\dots$ strictly grows lexicographically.
Establishing lexicographical growth of the sequence turns out to be quite technical, already in~\cite{DKZ19}; we believe that our argumentation is essentially the same, although conceptually simpler (and symmetrical) for the GKK algorithm.
The second step is an ingenious encoding into integers which exploits the symmetry to lower the obtained upper bound from the naive $2^n$ to $2^{n/2}$.

Step one relies on so-called alternating layers, which are defined with respect to minimal number of alternations between $\VMin$ and $\VMax$ for zero paths in $N^*$ towards $N$.
A similar result is derived in~\cite{GKK88} directly for the attracting layers, with a simpler proof.
It is required however for the second step to apply that nonzero integers appearing in the sequence alternate between positive and negative, which is not the case for attracting layers in general.
Assuming that the game is bipartite however (this incurs no loss of generality), one may combine the result of~\cite{GKK88} with the encoding of~\cite{DKZ19} and obtain the same result; here, we prefer to follow the two steps of~\cite{DKZ19} which allows to establish Theorem~\ref{thm:combinatorial_bound} in general.

\paragraph*{Step one: layers and their dynamics}

Again, we focus on $N^*$, but will later use the main result together with its dual to obtain the wanted bound.
Given a finite path $\pi:v_0 \re v_1 \re \dots \re v_{k}$ in $\game$ we define its number of alternations (towards $N$) $\alt(\pi) \in [0,\infty]$ to be the minimal $\ell$ such that there exist a decreasing sequence of $\ell+1$ indices $k \geq i_0 \geq i_1 \geq \dots \geq i_\ell$ such that
\begin{itemize}
\item $v_{i_0}, \dots, v_{k} \in N$,
\item for all $j \in [1,\ell]$, $v_{i_j},\dots,v_{i_{j-1}-1}$ all belong to $\VMax$ if $j$ is odd and to $\VMin$ if $j$ is even.
\end{itemize}

In particular a path has finite alternation number if and only if it ends in $N$ and it has alternation number 0 if and only if it is contained $N$.
Moreover note that a path from $v \notin N$ towards $N$ has even alternation number if and only if $v \in \VMin$.
The choice of the first layer being comprised of Max vertices is arbitrary, the proof below also goes through with the inverse convention.

We say that a path is zero if it visits only zero edges.
We define the {alternation depth $\alt(v)$ over vertices in $N^*$ by
\[
\alt(v) = \min\{\alt(\pi) \mid \pi \text{ is a zero path from $v$ to $N$ which remains in } N^*\}.
\]

An example is given in Figure~\ref{fig:alternating_layers}.
We say that a path from $v \in N^*$ is optimal if it is a zero path from $v$ to $N$ which remains in $N^*$ and achieves the above minimum.
Note that by definition of $N^*$, vertices in $N^*$ have a simple zero path towards $N$ hence $\alt(v)$ is finite and bounded by $n$.

\begin{figure}[ht]
\begin{center}
\includegraphics[width=0.9\linewidth]{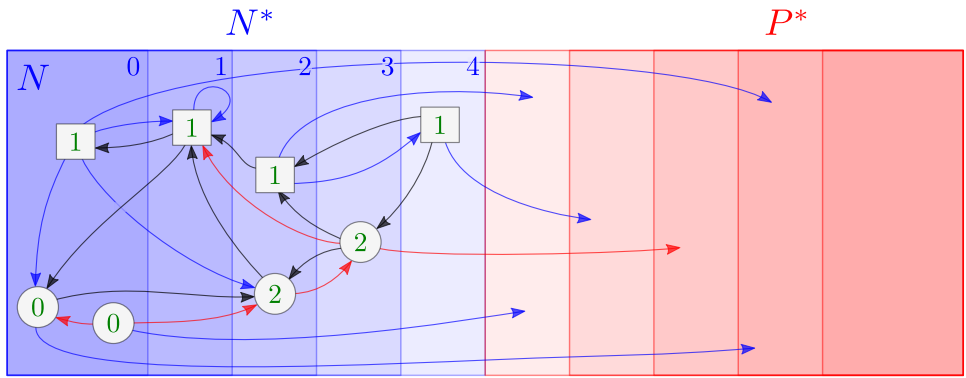}
\end{center}
\caption{The alternating layers, indicated by the green numbers, in the example of Figure~\ref{fig:gkk_partition}.
Notice that alternating layers (green numbers) and attractor layers (in blue) are completely different; however -- and quite surprisingly -- a close variant of Theorem~\ref{thm:strict_increase} holds for attractor layers (see~\cite{GKK88} for details).}
\label{fig:alternating_layers}

\end{figure}

We will study the dynamics of the sets
\[
A_i = \{v \in N^* \mid \alt(v)=i\}.
\]
We assume that the iteration is not over, $\delta < \infty$.
We use the notation $\game'$ for $\game^{\phi}$, where $\phi$ is the GKK potential and use primes for sets and quantities relative to $\game'$.
The following is the main result for the first step, it is proved in Appendix~\ref{app:strict_increase}.

\begin{theorem}
\label{thm:strict_increase}
If $N=N'$ and $P=P'$, then the sequence
\[
-|A_1|,|A_2|,-|A_3|,|A_4|, \dots
\]
strictly grows lexicographically.
\end{theorem}

\paragraph*{Step two: encoding into integers}

We now present the second step for the proof of Theorem~\ref{thm:combinatorial_bound}, due to~\cite{DKZ19}.
We let $k$ denote $|P|+|N|$, which can only decrease throughout the iteration thanks to Lemma~\ref{lem:evolution_n_p}.
Note that there exists $r \in [1,n-k]$ such that the layers $A_1,\dots,A_r$ are non-empty and $A_{r+1},A_{r+2},\dots$ are empty.
We let $s_r=1$ if $r$ is even and $0$ otherwise.

The argument relies on the following $n-k+1$-bit integer
\[
\alpha^- = \underbrace{0\dots 0}_{|A_1|} \underbrace{1 \dots 1}_{|A_2|} \underbrace{0 \dots 0}_{|A_3|} \dots \underbrace{s_r \dots s_r}_{|A_r|} 1 \underbrace{0 \dots 0}_{|P^*| - |P|},
\]
and its symmetric counterpart $\alpha^+$, which is defined in exactly the same way with respect to layers in $P^*$.

\begin{lemma}
If $k=k'$ then $\alpha^{'-} > \alpha^- + 2^{|P^*| - |P|}$ and likewise $\alpha^{'+} > \alpha^+ + 2^{|N^*| - |N|}$.
\end{lemma}

\begin{proof}
By Theorem~\ref{thm:strict_increase} the leftmost bit to switch from $\alpha^-$ to $\alpha^{'-}$ switches from $0$ to $1$, and occurs before the rightmost block of the form $1 0 \dots 0$ with $|P^*|-|P|$ zeros.
\end{proof}

We are finally ready to prove the announced bound.

\begin{proof}[Proof of Theorem~\ref{thm:combinatorial_bound}]
Consider $\alpha = \alpha^- + \alpha^+$, which is $\leq 2^{n-k+2}$.
Note that $|N^*|-|N| + |P^*|-|P| = n-k$, hence $\max(|N^*|-|N|,|P^*|-|P|)\geq \frac {n-k} 2$.
By the above lemma, if $k'=k$ then
\[
\alpha' > 2^{\max(|N^*|-|N|,|P^*|-|P|)} \geq 2^{\frac {n-k} 2}.
\]
Hence, there are at most $2^{n-k+2} / 2^{\frac {n-k} 2} = 4 . 2^{\frac {n-k} 2}$ consecutive iterations with the same $k$.
The bound follows since
\[
\sum_{k=0}^{n-1} 4 . 2^{\frac {n-k} 2} = O(2^{n/2}). \qedhere
\]
\end{proof}



\bibliography{bib}

\newpage

\appendix

\section{Proof of Theorem~\ref{thm:potential_reduction_theorem}}\label{app:potential_reduction_theorem}

This first appendix is devoted to the proof of Theorem~\ref{thm:potential_reduction_theorem}.
The first item directly follows from the fact that mean-payoffs of infinite paths in $\game$ and $\game^\phi$ are the same.
The third item follows from the second by symmetry; thus we focus on the second item.
We use the following result.

\begin{lemma}\label{lem:optimal_strat}
Let $\sigma_0$ be an $\Enp$-optimal Min strategy in $\game$ and $\pi = v_0 \re v_1 \re \dots \re v_k$ be a finite path consistent with $\sigma_0$ such that $\Enp_0(v_k)<\infty$.
Then we have $
\summ(\pi) \leq \Enp_0(v_0) - \Enp_0(v_k)
$.
\end{lemma}

\begin{proof}
Let $\pi'$ be an infinite path from $v_k$ consistent with $\sigma_0$ and such that $\Enp_0(v_k)=\Enp(w(\pi'))$.
Then $\pi \pi'$ is consistent with $\sigma_0$ thus $\Enp_0(v_0) \geq \Enp(w(\pi \pi'))$.
We thus obtain 
\[
\begin{array}{lcl}
\Enp(v_0) \ \ \geq \ \ \Enp(w(\pi\pi')) & = & \sup_{k' \geq 0}(\summ((\pi\pi')_{<k'}) \\ & \geq & \sup_{k' \geq k} (\summ((\pi \pi')_{<k'}) \\
& = & \summ(\pi) + \sup_{k' \geq 0} \summ(\pi'_{<k'}) \\
& = & \summ(\pi) + \Enp(w(\pi')) \ \  =  \ \ \summ(\pi) + \Enp_0(v_k). \qedhere
\end{array}
\]
\end{proof}

We now derive the wanted result.

\begin{proof}[Proof of second item of Theorem~\ref{thm:potential_reduction_theorem}]
Let $\phi: V \to \Z$ be a potential such that $0 \leq \phi \leq \Enp_0$; we aim to prove that $\Enp_0 = \phi + \Enp_\phi$ over $V$.
Over vertices with mean-payoff value $>0$ (which coincide over both games by the first item), both terms are infinite thanks to Corollary~\ref{cor:mean_payoffs_and_energies}.
Let $v$ be such a vertex with mean-payoff value $\leq 0$ (or equivalently, finite energy value).

Consider an $\Enp$-optimal Min strategy $\sigma_0: \VMin \to E$ in $\game$ and let $\pi=v_0 \re v_1 \re \dots$ be an infinite path consistent with $\sigma_0$.
Note that for any $k \geq 0$, $v_k$ has finite energy value, and thus we obtain thanks to Lemma~\ref{lem:optimal_strat}
\[
\begin{array}{lcl}
\summ_{\phi}(\pi_{<k}) & = & \summ(\pi_{<k}) + \phi(v_k) - \phi(v_0)\\
& \leq & \Enp_0(v_0)  \underbrace{- \Enp_0(v_k) + \phi(v_k)}_{\leq 0} - \phi(v_0) \ \ \leq \ \ \Enp_0(v_0) - \phi(v_0),
\end{array}
\]
hence $\Enp_{\phi}(v_0) \leq \sup_{\pi \const \sigma_0} \sup_{k \geq 0} \summ_\phi(\phi_{<k}) \leq \Enp_0(v_0) - \phi(v_0)$.

For the other inequality, consider an optimal Min strategy $\sigma_\phi$ in $\game_\phi$, and take $\pi \const \sigma_\phi$.
By applying Lemma~\ref{lem:optimal_strat} in $\game_\phi$ we now get
\[
\begin{array}{lcl}
\summ(\pi_{<k}) & = & \summ_\phi(\pi_{<k}) - \phi(v_k) + \phi(v_0)\\
& \leq & \Enp_\phi(v_0) - \underbrace{\Enp_\phi(v_k)}_{\geq 0} - \underbrace{\phi(v_k)}_{\geq 0} + \phi(v_0) \ \ \leq \ \ \Enp_\phi(v_0) + \phi(v_0),
\end{array}
\] 
and again the wanted result follows by taking a supremum.
\end{proof}

\section{Proof of Lemma~\ref{lem:evolution_n_p}}\label{app:evolution_n_p}

This small appendix is devoted to a proof of Lemma~\ref{lem:evolution_n_p}.
It states that the sets of $N$ and $P$ of negative and positive vertices can only decrease from an iteration to the next.

\begin{proof}[Proof of Lemma~\ref{lem:evolution_n_p}]
We let $\ext(v),\ext'(v) \in \Z$ denote the extremal weights of $v$ in $\game$ and $\game'$.
We prove that
\[
\begin{array}{ll}
\forall v \in N^*, \qquad& \ext(v) \leq \ext'(v) \leq 0 \\
\forall v \in P^*, \qquad& \ext(v) \geq \ext'(v) \geq 0.
\end{array}
\]
This implies the lemma: if $\ext'(v) < 0$ then necessarily $\ext(v)<0$ so $v$, therefore $N' \subseteq N$; likewise, $P' \subseteq P$.
We only prove the first line since the second follows by symmetry.

For the left inequality it suffices to observe that the weight of edges outgoing from $N^*$ can only increase: edges pointing to $N^*$ keep the same weight while those pointing towards $P^*$ are increased by $\delta$.
For the inequality on the right we make a quick case disjunction.
\begin{itemize}
\item Let $v \in N^* \cap \VMax$.
Then all extremal edges are non-positive, and those which point towards $P^*$ are even $\leq -\delta$ by definition of $\delta$ hence they all remain non-positive.
\item Let $v \in N^* \cap \VMin$.
The result follows directly if $v$ has a non-positive outgoing edge towards $N^*$ since it is left unchanged.
Otherwise $v \in SN$ hence $v$ has an outgoing edge of weight $\leq - \delta$ which therefore remains non-positive. \qedhere
\end{itemize} 
\end{proof}

\section{Proof of Theorem~\ref{thm:strict_increase}}\label{app:strict_increase}

This appendix is devoted to the proof of Theorem~\ref{thm:strict_increase}, which is the most technical one in the paper.

Towards proving the theorem, we define two relevant indices $i_D$ and $i_A$ which we respectively call the departure index and arrival index.
As their names suggest the first is relevant to vertices which leave $N^*$, that is, those in $N^* \cap P^{'*}$, while the second is relevant to arriving vertices, those in $P^* \cap N^{'*}$.
We let
\[
\begin{array}{lcl}
i_D & = & \min\{\alt(v) \mid v \in \VMax \cap N^* \text{ and } v \in P^{'*}\}, \\
i_A &=& \min\{\alt(e_0\pi_1) \mid e_0 \in E \cap [(P^* \cap \VMin) \times  N^*], w(e_0)=\delta, \text{ and } \pi_1 \text{ is optimal from } v_1 \text{ in } \game\}
\end{array}
\]
Note that if finite, $i_D$ is odd and $i_A$ is even.
We now provide a sequence of incremental results that eventually give the theorem.

\begin{lemma}
Assume that $N'=N$ and $P'=P$.
\begin{enumerate}[(i)]
\item \label{it1} For all $v \in N^*$, if $\alt(v) < i_D$ then $v \in N^{'*}$ and $\alt'(v) \leq \alt(v)$.
\item \label{it2} Any path $\pi'$ which is optimal in $\game'$ but is not zero in $\game$ satisfies $\alt'(\pi') \geq i_A$.
\item \label{it3} For all $v \in P^* \cap N^{'*}$ it holds that $\alt'(v) \geq i_A$.
\item \label{it4} For all $i \leq \min(i_D -1,i_A)$ we have $A_i \subseteq A'_i$ and for all $i \leq \min(i_D,i_A -1)$ we have $A'_i \subseteq A_i$.
\item \label{it5} If $i_A < i_D$ then $|A'_{i_A}|>|A_{i_A}|$. 
\item \label{it6} We have $\delta_{E}^- < - \delta$ and likewise $\delta_{A}^+ > \delta$.
\item \label{it7} If $i_A = \infty$ then $i_D < \infty$.
\item \label{it8} Theorem~\ref{thm:strict_increase} holds.
\end{enumerate}
\end{lemma}

Items (\ref{it1}), (\ref{it2}) and (\ref{it3}) build towards item (\ref{it4}) which is the main intermediate result.
Items (\ref{it5}) and (\ref{it7}) have a similar proof although (\ref{it7}) also relies on (\ref{it6}), and build up to the conclusion.

\begin{proof}
\begin{enumerate}[(i)]
\item We prove the claim by induction on the length $k$ of the smallest optimal path $\pi=v_0 \re \dots \re v_k$ from $v_0=v$.
Note that $\pi$ is zero in $\game$ and remains in $N^*$ hence it is also zero in $\game'$.
If $\pi$ has length zero then $v \in N = N'$ hence $v \in N^{'*}$ and $\alt'(v)=0 \leq \alt(v)$, so we now assume $k > 0$ and that the result is known for vertices with an optimal path of length $\leq k -1$.

It holds by induction that $v_1, v_2, \dots, v_{k} \in N^{'*}$ hence it suffices to prove that $v \in N^{'*}$ since it implies that $\pi$ is a zero path in $\game'$ which remains in $N^{'*}$.
If $v \in \VMin$ then $v$ has a zero edge in $\game'$ towards $v' \in N^{'*}$ hence $v \in N^{'*}$.
Otherwise it holds that $v \in N^{'*}$ because $v \in P^{'*}$ would contradict that $\alt(v) < i_D$.

\item Let $\pi':v_0 \re \dots \re v_{k}$ be such a path.
It cannot be that $\pi'$ is included in $N^*$ otherwise it would be zero in $\game$, and we let $i_0$ be the largest index such that $v_{i_0} \in P^*$.
Since $w'(v_{i_0}v_{i_0+1})=0$ we have $w(v_{i_0}v_{i_0+1})=\delta >0$ hence it must be that $v_{i_0} \in \VMin$ otherwise we would have $v_{i_0} \in P=P'$ which contradicts that $v_{i_0} \in N^{'*}$.
We now let $\pi$ be an optimal path from $v_{i_0+1}$.
Then we have $\alt(\pi') \geq \alt((v_{i_0} v_{i_0+1}) \pi) \geq i_A$.

\item Let $v \in P^* \cap N^{'*}$ and let $\pi':v_0 \re \dots \re v_{k}$ be an optimal path from $v=v_0$ in $\game'$.
We assume for contradiction that $\alt'(v) < i_A$, which thanks to the previous item implies that $\pi'$ is zero in $\game$.
Since $v_{k} \in N' = N$ and $v=v_0 \in P^*$ there is an index $i_0$ such that $v_{i_0} \in P^* $ and $v_{i_0+1} \in N^*$.
This contradicts the fact that $w(v_{i_0}v_{i_0+1})=w'(v_{i_0}v_{i_0+1})=0$.

\item We prove the two results together by induction on $i$.
For $i=0$ we have $A_0=N=N'=A'_0$ hence we let $i \geq 1$ and assume that both results hence the equality are known for smaller values.

By item (\ref{it1}) if $i < i_D$ and $v \in A_i$ then $v \in N^{'*}$ and $\alt'(v) \leq i$, but our induction hypothesis tells us that $\alt'(v)$ cannot be $<i$ hence $A_i \subseteq A'_{i}$.

Conversely let $v \in A'_i$, assume $i < i_A$, let $\pi:v_0 \re \dots \re v_k$ be an optimal path from $v_0=v$ in $\game'$, and let $j_0>0$ be the smallest index such that $v_{j_0} \notin A'_i$.
We assume that $\pi$ is chosen such that $v_{j_0}$ is minimal, and prove the result by an inner induction on $j_0$.
Since $\alt'(v) = i < i_A$ we know by item (\ref{it3}) that $v \in N^*$.

If $j_0=1$, that is if $v_1 \in A'_{i-1}$, then thanks to the (outer) induction hypothesis for all $j \geq 1$ we have $v_j \in A'_{k_j} = A_{k_j}$ for some $k_j < i$, hence for all $j \geq 0$ we have $v_j \in N^*$.
Hence $\pi'$ remains in $N^*$ and is zero in $\game'$ thus it is also zero in $\game$ and $\alt(v) \leq \alt(\pi)=i$.
We conclude thanks to the (outer) induction that $\alt(v)=i$.

If $j_0 \geq 2$ then the inner induction hypothesis gives $v_j \in N^*$ for $j \in [1,j_0]$ and the outer induction hypothesis gives $v_j \in N^*$ for $j \in [j_0+1,k]$, and we repeat the same argument.

\item Assume that $i_A < i_D$.
By item (\ref{it4}) it holds that $A_{i_A} \subseteq A'_{i_A}$ hence it suffices to find $v_0 \in A'_{i_A} \setminus A_{i_A}$ and we take $v_0$ given by the definition of $i_A$:
$v_0 \in P^* \cap \VMin$ is such that there is $v_1 \in N^*$ with $v_0v_1 \in E$, $w(v_0v_1)=\delta$ (which implies $w'(v_0v_1)=0$) and $\alt(v_1) \leq i_A$.

Again by (\ref{it4}) it holds that $v_1 \in N^{'*}$ and $\alt'(v_1) \leq \alt(v_1)$, thus $v_0 \in \VMin$ has a zero edge in $\game'$ towards a vertex of $N^{'*}$ and therefore $v_0 \in N^{'*}$.
Now $\alt'(v_0) \leq \alt((v_0 v_1) \pi'_1)$, where $\pi'_1$ is an optimal path from $v_1$ in $\game'$ and hence $\alt'(v_0) \leq i_A$.
Yet again thanks to (\ref{it4}) it cannot be that $\alt'(v_0)<i_A$ since $A'_{i} \subseteq A_i$ for $i < i_A$ and $v_0 \in P^*$, therefore we conclude that $v_0 \in A'_i$.

\item Assume for contradiction that $\delta_{N,2} = - \delta$.
Then there is $v \in SN$ such that $\ext(v)=-\delta$, hence $\ext'(v)=0$ which contradicts $N'=N$.
The proof of the second statement is symmetric.

\item If $i_A=\infty$ then there is no edge with weight $\delta$ in $\game$ from $P^* \cap \VMin$ to $N^*$ hence $\delta_{P,1} > \delta$ therefore it must be by item (\ref{it6}) that $\delta = - \delta_{A}^-$.
We let $e_0=v_0v_1$ be an edge with weight $-\delta$ from $v_0 \in \VMax \cap N^*$ to $v_1 \in P^*$.
We claim that $P^{'*} \supseteq P^*$ which proves the result since then $v_0 \in \VMax$ has an edge $e_0$ which is zero (hence non-negative) in $\game'$ towards $P^{'*}$, hence $v_0 \in P^{'*}$ and $i_D \leq \alt(v_0)$.

This follows from a quick induction over attractor-layers towards $P=P'$ over zero edges in $\game$:
a vertex $v \in \VMax \cap (P^* \setminus P)$ has a zero edge, which remains zero, in $\game$ towards a vertex in the previous layer, and by assumption vertices $v \in \VMin \cap (P^* \setminus P)$ have all their edges towards $N^*$ which are $\geq \delta_{P,1}> \delta$ hence remain positive.

\item By item (\ref{it7}) $m=\min(i_D,i_A)$ is finite, and by item (\ref{it4}) we have $A_i=A'_i$ for all $i < m$.
If $m = i_A \leq i_D-1$ then moreover $A_m \subseteq A'_m$ and the inclusion is strict by item (\ref{it5}), which concludes.
Otherwise $m = i_D \leq i_A-1$ hence $A_m \supseteq A'_m$ and the inclusion is strict by definition of $i_D$. \qedhere
\end{enumerate}
\end{proof}

Even broken in elementary steps the proof above remains very tedious, we are not aware unfortunately of simplifications that could be made.

\end{document}